%% file: main.tex
\documentclass[letterpaper, 10 pt, conference]{ieeeconf} 
\IEEEoverridecommandlockouts
\usepackage{amsmath,amssymb}
\usepackage{etoolbox}
\usepackage{graphicx}
\usepackage{bm}
\usepackage{svg}
\usepackage{textcomp}
\usepackage{xcolor}
\usepackage{cite}
\usepackage{url}
\usepackage[hidelinks]{hyperref}
\usepackage{algorithm}
\usepackage[noend]{algorithmic}

\usepackage{amsthm}
\input{settings}

\usepackage{xprintlen}
\usepackage{cases}
\usepackage{tabularx,ragged2e}
\usepackage{booktabs}
\usepackage{nccmath}
\usepackage{caption}
\usepackage{float}
\usepackage{subcaption}
\usepackage{mathtools}
\usepackage{tikz}
\usepackage{caption}  
\usetikzlibrary{matrix, arrows.meta}
\usetikzlibrary{automata, positioning}

\def\BibTeX{{\rm B\kern-.05em{\sc i\kern-.025em b}\kern-.08em
    T\kern-.1667em\lower.7ex\hbox{E}\kern-.125emX}}
\usepackage{array}
\newcolumntype{P}[1]{>{\centering\arraybackslash}p{#1}}

\begin{document}

\title{A Reactive Redistribution Mechanism for STL Tasks in Multi-Agent Systems Under Time-Varying Communication}

\author{Gregorio Marchesini, Bjarne Jan Jesse Moro, Siyuan Liu, 
Lars Lindemann and Dimos V. Dimarogonas
\thanks{This work was supported in part by the Horizon Europe EIC project SymAware (101070802), ERC LEAFHOUND Project, Swedish Research Council (VR), and the Knut and Alice Wallenberg (KAW) Foundation.}
\thanks{G. Marchesini, B. J. J. Moro, and D. V. Dimarogonas are with the Division of Decision and Control Systems, KTH Royal Institute of Technology, Stockholm, Sweden.
	E-mail: {\tt\small \{gremar,bjjmoro,dimos\}@kth.se}. S. Liu is with the Department of Electrical Engineering, Eindhoven University of Technology, the Netherlands. 
    E-mail: {\tt s.liu5@tue.nl}.  
    L. Lindemann is with the Automatic Control Laboratory, ETH Zürich, Zürich, Switzerland.
	E-mail: {\tt llindemann@ethz.ch}.  
}
}

\maketitle
\begin{abstract}
\input{content/abstract}
\end{abstract}


\section{Introduction}
\input{content/introduction}\par

\input{content/notation}
\section{Preliminaries}\label{sec:preliminaries}
\input{content/preliminaries}

\section{Problem Formulation}\label{sec:problem}
\input{content/problem}

\section{Graph Transition System}
\input{content/gts}\label{sec:gts}
\section{Task Decomposition over GTS}\label{sec:decomposition}
\input{content/task_dec}

\section{Overall Algorithmic Implementation}\label{sec:implementation}

\input{content/implementation}

\section{Numerical Experiments}\label{sec:numerical}
\input{content/numerical_studies}
\section{Conclusion}\label{sec:conclusion}
\input{content/conclusion}
\vspace{-0.25cm}
\appendix
\input{content/appendix}


\bibliographystyle{IEEEtran}
\bibliography{references} 

\end{document}

%% file: settings.tex
\renewcommand{\vec}[1]{\bm{#1}}
\newcommand{\true}{\top}

\newtheoremstyle{compact}
{0pt}                
{0pt}                
{\itshape}           
{}                   
{\bfseries}          
{.}                  
{.5em}               
{}                   

\theoremstyle{compact} 

\newtheorem{proposition}{\bf{Proposition}}
\newtheorem{definition}{\bf{Definition}}
\newtheorem{lemma}{\bf{Lemma}}
\newtheorem{theorem}{\bf{Theorem}}
\newtheorem{remark}{\bf{Remark}}
\newtheorem{assumption}{\bf{Assumption}}

\newtheorem{corollary}{\bf{Corollary}}
\newtheorem{example}{\textbf{Example}}

\def\myunderbar#1{\underline{\sbox\tw@{$#1$}\dp\tw@\z@\box\tw@}}

%% file: content/abstract.tex
We present a communication-aware task decomposition framework for multi-agent systems with collaborative relative configuration objectives specified in Signal Temporal Logic (STL), allowing for dynamic task reallocation under time-varying communication networks. Building on our prior work, the framework supports the direct use of existing feedback controllers for reactive task satisfaction. We address two key challenges: disjunctive STL specifications and time-varying communication networks. Disjunctive specifications are handled through a graph transition system that captures the alternative task sequences induced by logical OR operators. To address time-varying connectivity, we introduce a redistribution mechanism that transfers tasks from disconnected agents to connected ones as the network evolves while preserving decentralized execution. Simulations and experiments on a swarm of Crazyflie drones demonstrate scalability in the number of agents, communication connectivity, and specification complexity.

%% file: content/introduction.tex
A broad class of multi-agent system (MAS) applications requires enforcing prescribed configurations (formations) among agents, such as drone light shows \cite{droneshow} or search and rescue missions \cite{searchandrescue}. Classical MAS control has been extensively studied under static specifications; however, emerging applications increasingly require the satisfaction of time-varying collaborative objectives involving explicit temporal constraints. Signal Temporal Logic (STL) \cite{maler2004monitoring} has become a standard formalism for specifying such spatio-temporal behaviors, as it allows for both formal verification and controller synthesis within a unified framework. Existing methods for enforcing STL specifications in MAS can be broadly categorized into open-loop and feedback-based approaches. Open-loop approaches translate STL specifications into mixed-integer linear constraints, yielding trajectory synthesis problems that can be solved via mixed-integer programming \cite{planningmulti,buyukkoccak2021distributed}. While these methods are sound and complete, their computational complexity scales poorly with the number of agents and specifications, generally restricting their applicability to small-scale systems. In contrast, feedback-based approaches address scalability limitations by embedding STL constraints into differentiable, time-varying inequality constraints, commonly via Control Barrier Functions \cite{lindemann2025formal} or Prescribed Performance Funnels \cite{long2024dynamic}. This reformulation enables the synthesis of feedback controllers either analytically or by solving convex quadratic programs. Compared to open-loop approaches, these methods improve scalability with the number of agents, exhibit inherent robustness due to their reactive nature, and handle a rich fragment of STL specifications, at the expense of formal completeness, i.e., they may fail to find a solution even when one exists.

\begin{figure}
    \centering
    \includegraphics[width=\linewidth]{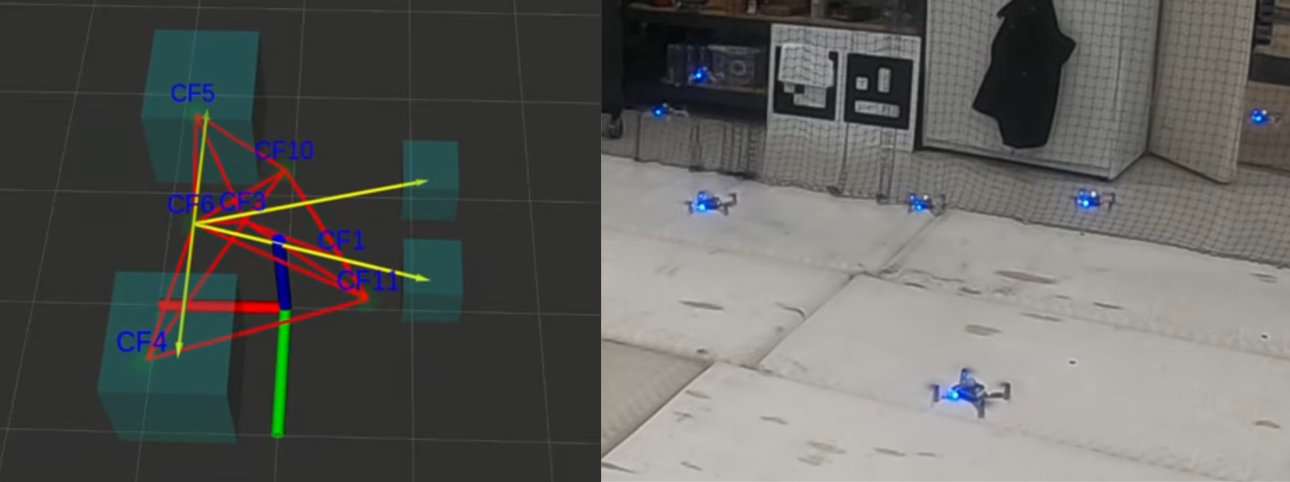}
    \caption{Crazyflie drones achieving a relative configuration task under time-varying communication. On the left, the edges of the communication graph are shown in red, while the required relative configurations are shown as cyan hyper-rectangles with yellow arrows representing their relative configuation. Full video available at \url{https://youtu.be/qLjHKBQRmYw}}
    \label{fig:main figure}
\end{figure}

A fundamental yet often implicit assumption in feedback-based multi-agent STL control is the compatibility of task dependencies with the underlying communication topology. Specifically, when a specification couples multiple agents, those agents must be able to exchange state information to enforce the associated constraints. In practical settings, however, communication is typically limited by distance-based constraints, and therefore, task-coupled agents may not be directly connected. This mismatch renders certain specifications unrealizable under purely local feedback laws and has been typically addressed by two main approaches: distributed state estimation \cite{zaccherini2026STL} and task decomposition \cite{me}. Among these methods, task decomposition reformulates tasks defined over non-communicating agents into a set of specifications over subsets of agents connected in the communication graph that imply the original task.

In this work, we propose a task decomposition framework, building upon our results in \cite{me}, to satisfy spatio-temporal relative configuration constraints, expressed as STL formulae, under time-varying communication graphs. Our contributions compared to \cite{me} are as follows: first, we consider task decomposition over disjunctive type specifications, thereby enlarging the class of specifications that can be addressed. Second, we consider \emph{dynamic} communication graphs among the agents, which naturally arise in real-world practical scenarios. Third, we provide a scalability analysis with respect to the number of agents, the complexity of the STL task, and the connectivity of the communication graph. Finally, simulation studies are complemented by experimental validation on a swarm of Crazyflie drones (see Fig. \ref{fig:main figure})

\subsubsection*{Organization} Sec.~\ref{sec:preliminaries} introduces the required preliminaries on STL for the specification of collaborative relative spatiotemporal tasks. Sec.~\ref{sec:problem} presents the problem formulation. Secs.~\ref{sec:gts} and~\ref{sec:decomposition} develop the graph transition system and task decomposition as the main ingredients for satisfying relative STL tasks under time-varying range-limited communication. The overall integrated algorithmic framework is presented in Sec.~\ref{sec:implementation}, while experimental validation and conclusions are provided in Sec.~\ref{sec:numerical} and~\ref{sec:conclusion}, respectively.

\subsubsection*{Notation} Real and non-negative real numbers are indicated as $\mathbb{R}$ and $\mathbb{R}_{\geq 0}$, respectively. Column vectors are indicated in bold, and $\vec{x}[k]$ is the $k$-th entry of $\vec{x}$. The matrix $\text{Id}_{n} \in \mathbb{R}^{n \times n}$ is the identity matrix of size $n$. For $\vec{x},\vec{y} \in \mathbb{R}^n$, the Hadamard product is $\vec{x}\odot \vec{y} = [\vec{x}[1]\cdot \vec{y}[1], \ldots \vec{x}[n]\cdot \vec{y}[n]]^T$. For a vector $\vec{x}$, $\|\vec{x}\|$ is the standard 2-norm. The notation $|\mathcal{C}|$ denotes the cardinality of the set $\mathcal{C}$, while $\mathcal{A}\oplus \mathcal{B} = \{\vec{a} + \vec{b} \mid \forall \vec{a} \in \mathcal{A}, \vec{b} \in \mathcal{B}\}$ is the Minkowski sum of $\mathcal{A} \subseteq \mathbb{R}^{n}$ and $\mathcal{B}\subseteq \mathbb{R}^{n}$. The set $2^{\mathcal{A}}$ is the power set of $\mathcal{A}$, i.e., the set of all subsets of $\mathcal{A}$. For vectors $\vec{a}_i \in \mathbb{R}^{n_i}$ with index $i$ in the set $\mathcal{K}$, the vector $\vec{a} = [\vec{a}_i]_{i\in \mathcal{K}} \in \mathbb{R}^{\sum_{i\in \mathcal{K}}n_i}$ is the stacked column vector of all vectors $\vec{a}_i$.

%% file: content/preliminaries.tex
Consider a set of agents with indices in $\mathcal{V} = \{1, \ldots N\}$, where each agent follows a nonlinear input-affine dynamics 
\begin{equation}\label{eq:main dynamics}
\begin{aligned}
\dot{\vec{x}}_i(t) &= f_i(\vec{x}_i(t)) + g_i(\vec{x}_i(t))\vec{u}_i(t),\\
\end{aligned}
\end{equation}
such that $\vec{x}_i \in \mathbb{R}^{n_x}$ and  $\vec{u}_i \in \mathbb{U} \subset \mathbb{R}^{n_u}$ are the state and control input of agent $i$, respectively. We consider $\mathbb{U} \subset \mathbb{R}^{n_u}$ to be compact, while $f_i : \mathbb{R}^{n_x} \rightarrow \mathbb{R}^{n_x}$ and $g_i : \mathbb{R}^{n_x} \rightarrow \mathbb{R}^{n_u}$ are locally Lipschitz continuous in their arguments. Given the system's state $\vec{x} = [\vec{x}_i]_{i \in \mathcal{V}} \in \mathbb{R}^{N n_x}$, and assuming each agent is subject to a measurable input signal $\vec{u}_i : [t_0,t_1] \rightarrow \mathbb{U}$, starting from the initial state $\vec{x}_0$, we denote by  $\vec{x}(t) = q_{\text{sol}}(t;\vec{x}_0, \{\vec{u}_i\}_{i\in \mathcal{V}}),\; \forall t\in [t_0,t_1]$ the unique maximal state solution\footnote{Without loss of generality we assume in this work that all solutions to \eqref{eq:main dynamics} under a measurable and bounded input signal $\vec{u}_i : [t_0,t_1] \rightarrow \mathbb{U}$ are defined over the whole interval $[t_0,t_1]$.} to \eqref{eq:main dynamics}, where $\vec{x}(t_0) = \vec{x}_0$.

We define the \textit{relative state} between agents $i,j \in \mathcal{V}$ as $\vec{e}_{ij} := \vec{x}_j - \vec{x}_i$, $\vec{e}_{ij} \in \mathbb{R}^{n_x}$. Moreover, $\vec{p}_i = S\vec{x}_i$ is the position of agent $i$, where $S$ is an appropriate selection matrix\footnote{For given input size $n$ and output size $m$, a selection matrix $S \in \{0,1\}^{n\times m}$ satisfies $\sum_{j=1}^n S[i,j] = 1$ for all $i$, such that in $\vec{y} = S\vec{x}$, the vector $\vec{y}$ represents a selection of the elements of $\vec{x}$.}, such that $\vec{p}_{ij} := \vec{p}_j -\vec{p}_i = S \vec{e}_{ij}$ is the relative position of $i$ and $j$. A \textit{static} \textit{undirected} graph $\mathcal{G}(\mathcal{V},\mathcal{E}) \in \Gamma$ is a tuple of the agents' indices $\mathcal{V}$ (vertices of the graph) and a set of edges $\mathcal{E} \subseteq \mathcal{V} \times \mathcal{V}$, such that $(i,j) \in \mathcal{E} \Leftrightarrow (j,i) \in \mathcal{E}$. The set $\Gamma$ is the set of undirected graphs on $\mathcal{V}$. The set of \textit{neighbours} of agent $i$ is defined as $\mathcal{N}_i := \{j \in \mathcal{V} \mid (i,j) \in \mathcal{E}\}$. When $\mathcal{E}$ is clear or irrelevant, we adopt the shorthand notation $\mathcal{G} := \mathcal{G}(\mathcal{V}, \mathcal{E})$, and we write $(i,j) \in \mathcal{G}$ for $(i,j) \in \mathcal{E}$.

A directed path $\vec{\pi}_{ij}$ between agents $i$ and $j$ in $\mathcal{G}$ is defined as a sequence of $L$ non-repeated indices in $\mathcal{V}$ such that $
\vec{\pi}_{ij} := k_1 \; k_2\; k_3\; \ldots k_L$. Namely, letting $\vec{\pi}_{ij}[k] \in \mathcal{V}$ indicate the $k$-th element of $\vec{\pi}_{ij}$ we have  $\vec{\pi}_{ij}[1] = i$, $\vec{\pi}_{ij}[L] = j$, $(\vec{\pi}_{ij}[k],\vec{\pi}_{ij}[k+1]) \in \mathcal{E}$. The function $\epsilon : \vec{\pi}_{ij} \mapsto \{(\vec{\pi}_{ij}[k],\vec{\pi}_{ij}[k+1])\}_{k=1}^{L-1}$ associates to each path $\vec{\pi}_{ij}$ the set of consecutive edges along the path. Similarly, a \textit{cycle} $\vec{\pi}_{\omega}$ is defined as a sequence of non-repeated indices, except the initial and final nodes, i.e., $\vec{\pi}_{\omega}[1] = \vec{\pi}_{\omega}[L]$.

Beyond static graphs, we define a \textit{dynamic} graph as a function $\mathcal{G} : \mathbb{R}_{\geq 0} \rightarrow \Gamma$, $t \mapsto \mathcal{G}(\mathcal{V},\mathcal{E}(t))$, where $\mathcal{E}: \mathbb{R}_{\geq 0} \rightarrow 2^{\mathcal{V} \times \mathcal{V}}$ is a function of time. In particular, we define the dynamic \textit{communication graph} $\mathcal{G}_c( \mathcal{V}, \mathcal{E}_c(t))$ with time-varying edges $\mathcal{E}_c(t) = \{ (i,j) \in \mathcal{V}\times \mathcal{V} \mid \| \vec{p}_{ij}(t) \| \leq r \}$ and neighbouring set $\mathcal{N}_{i,c}(t)$ for each $i\in \mathcal{V}$.

\subsection{Relative Configurations and Signal Temporal Logic}
For a \textit{desired relative vector} $\vec{c}_{ij} \in \mathbb{R}^{n_x}$ and a given \textit{size vector} $\vec{\alpha}_{ij} \in \mathbb{R}^{n_x}_{\geq 0}$, we defined a relative configuration $\mathcal{H}_{ij}(\vec{\eta}_{ij}) \subset \mathbb{R}^{n_x}$ between $i$ and $j$, with $\vec{\eta}_{ij} = [\vec{c}_{ij}^T, \vec{\alpha}_{ij}^T]^T \in \mathbb{R}^{n_x}\times \mathbb{R}^{n_x}_{\geq 0}$, as a shifted hyper-rectangle:
\begin{equation}\label{eq:predicate level set}
\begin{aligned}
&\mathcal{H}_{ij}(\vec{\eta}_{ij}) :=  
\{ \vec{e}_{ij}  \mid  |\vec{e}_{ij}[k] - \vec{c}_{ij}[k] | \leq \vec{\alpha}_{ij}[k],\\ &\; \forall k = 1, \ldots n_{x}\} = 
\{ \vec{e}_{ij}  \mid  H(\vec{e}_{ij} - \vec{c}_{ij}) - M\vec{\alpha}_{ij} \leq \vec{0}\}
\end{aligned}
\end{equation}
where $H = [\text{Id}_{n_x}, -\text{Id}_{n_x}]^T$, $M = [\text{Id}_{n_x},\text{Id}_{n_x}]^T$. 
Notably, hyper-rectangles are closed under Minkowski summation, i.e., the Minkowski sums of hyper-rectangles are themselves hyper-rectangles. In particular, given a graph $\mathcal G$ and a path of agents $\vec{\pi}_{ij}$ in $\mathcal{G}$, where each pair $(r,s)\in \epsilon(\vec{\pi}_{ij})$ is subject to  $\vec{e}_{rs} \in \mathcal{H}_{rs}(\vec{\eta}_{rs})$, then it holds \cite[Prop. 3]{me}
\begin{equation}\label{eq:minkowsky inclusion to linear constraints}
\bigoplus_{(r,s) \in \epsilon(\vec{\pi}_{ij})} \mathcal{H}_{rs}(\vec{\eta}_{rs}) = \mathcal{H}_{ij}\left(\sum_{(r,s) \in \epsilon(\vec{\pi}_{ij})} \vec{\eta}_{rs}\right).
\end{equation}

For a given relative configuration $\mathcal{H}_{ij}(\vec{\eta}_{ij})$ we define a \textit{boolean predicate} as 
\begin{equation}\label{eq:predicate definition}
\mu_{ij} = \mu_{ij}(\vec{e}_{ij}) :=
\left\{
\begin{array}{@{}l@{}}
\top \;\; \vec{e}_{ij} \in \mathcal{H}_{ij}(\vec{\eta}_{ij}),\\
\bot \;\; \text{otherwise}
\end{array}
\right. 
\end{equation}
and, for a subset $\mathcal{E} \subseteq \mathcal{V}\times \mathcal{V}$ of agents' pairs, with stacked relative states $\vec{e}_{\mathcal{E}} = [\vec{e}_{rs}]_{(r,s) \in \mathcal{E}}$, we define $\theta(\vec{e}_{\mathcal{E}}) = \bigwedge_{(r,s) \in \mathcal{E}} \mu_{rs}$, such that 
$\theta = \top \Leftrightarrow \vec{e}_{rs} \in \mathcal{H}_{rs}(\vec{\eta}_{rs}),\; \forall(r,s)\in\mathcal{E}$. In other words, $\theta$ encodes a Boolean satisfaction (True or False) of a set of relative configurations of the agents' pairs in $\mathcal{E}$. By this formalism, we can represent time-constrained relative configurations among agents using STL formulas as:
\begin{equation}\label{eq:working fragment}
 \varphi      := G_{I}\theta  \mid  F_{I}\theta,  \qquad \psi        := \varphi \mid \psi_1 \land \psi_2 \mid  \psi_1 \lor \psi_2,
\end{equation}
 where the operators $G_{I}$ and $F_{I}$ represent the \textit{always} and \textit{eventually} temporal operators over the time interval $I \subset \mathbb{R}_{\geq0}$, while the operators $\land$ and $\lor$ represent the logical conjunction and disjunction, respectively. We will denote by $\Psi$ the set of formulas in the fragment \eqref{eq:working fragment}. For notational convenience, we hereafter denote by $\psi^{\land}$ formulas that \textit{only} contain conjunctions as 
 $\psi^{\land} := \varphi_1 \land \varphi_2 \land \ldots = \bigwedge_{k} \varphi_k$, 
and we use the notation $\mu_{ij} \in \psi$ to denote that a predicate $\mu_{ij}$ is \textit{contained} in the nested formula $\psi$. For a signal $\vec{x} : \mathbb{R}_{\geq 0} \rightarrow \mathbb{R}^{N n_x}$ describing the state evolution of the multi-agent system under the dynamics in \eqref{eq:main dynamics}, the following semantic rules specify when each formula is satisfied 
\begin{equation}\label{eq:semantics}
\begin{aligned}
&\vec{x} \models F_I \theta
\Leftrightarrow \exists\tau\in I = [a,b]: \theta(\vec{e}_{\mathcal{E}}(\tau))= \top;\\
&\vec{x}\models G_I\theta
\Leftrightarrow \forall\tau\in I = [a,b]:\theta(\vec{e}_{\mathcal{E}}(\tau)) = \top;\\
&\vec{x}\models \bigcirc_k \varphi_k
\Leftrightarrow \bigcirc_k\,\vec{x}\models\varphi_k,
\quad \bigcirc\in\{\vee,\wedge\},
\end{aligned}
\end{equation}
where $\vec{x} \models \square$ denotes that the state signal $\vec{x} : \mathbb{R}_{\geq 0} \rightarrow \mathbb{R}^{Nn_x}$ satisfies a formula $\square$. While the intervals $I$ can potentially take any value, we make the following  assumption:
\begin{assumption}\label{ass:periodicity} Let $\mathcal{I} = \{ I_{p_i} \subset \mathbb{R}_{\geq 0} \}_{p_1=1}^{p_{n_p}}$, with $n_p \geq 1$, be an ordered collection of disjoint intervals $I_{p_i} = [a_{p_i}, b_{p_i}]$ such that $p_{i} < p_{j} \Rightarrow b_i < a_j$. Each task $\varphi$, as defined in \eqref{eq:working fragment}, is of the form $\varphi = T_{I_{p_i}}\,\theta$, where $T \in \{G, F\}$ and $I_{p_i} \in \mathcal{I}$. \end{assumption}
In practice, Assumption~\ref{ass:periodicity} implies that a formula $\psi$ consists of conjunctions/disjunctions of subformulas $\varphi$, sequentially ordered over the intervals in $\mathcal{I}$. While this may not always hold, it reflects standard engineering practice where a high-level planner schedules tasks over $\mathcal{I}$.
For a general task $\psi$, we define the \textit{task horizon} and \textit{task activation} as $t_{\text{hr}}(\psi) = \max_{p_i} b_{p_i}$ and $t_{\text{ac}}(\psi) = \min_{p_i} a_{p_i}$, such that, by the semantics in \eqref{eq:semantics}, values of $\vec{x}(t_{\text{ac}}(\psi) - \tau)$ and $\vec{x}(t_{\text{hr}}(\psi) + \tau)$, for any $\tau \geq 0$, do \textit{not} affect the satisfaction of $\psi$.

\subsection{Task Graphs}
To coherently capture both task and communication coupling among agents, we represent a task $\psi$ as a \textit{task graph}, which will enable direct comparison with $\mathcal{G}_c(t)$ and reveal structural mismatches that may hinder the satisfaction of $\psi$ via feedback, as analyzed shortly in Sec. \ref{sec:problem}. Specifically, let the graph mapping $\mathfrak{G} : \Psi \rightarrow \Gamma$, defined as
\begin{equation}\label{eq:task mapping}
\mathfrak{G}(\psi) = \mathcal{G}(\mathcal{V}, \mathcal{E}_{\psi}); \; \mathcal{E}_{\psi} = \{ (i,j) \in \mathcal{V}\times \mathcal{V} \mid \exists \mu_{ij} \in  \psi\}.
\end{equation}
In other words, $\mathfrak{G}(\psi)$ associates a task $\psi$ with a corresponding static graph $\mathcal{G}(\mathcal{V}, \mathcal{E}_{\psi})$, a \textit{task graph}, that defines the task couplings among agents. Practically, two agents are coupled in $\mathfrak{G}(\psi)$ if there exists a relative configuration predicate $\mu_{ij}$ defined over their relative state $\vec{e}_{ij}$ in $\psi$. We denote by $\mathfrak{E}(\psi)$ the set of edges of the task graph $\mathfrak{G}(\psi)$, and by $\mathfrak{N}_{i}(\psi)$ the set of neighbors of agent $i$ in $\mathfrak{G}(\psi)$. For a given agent $i$, we define the state $\vec{x}_{\mathfrak{N}_{i}(\psi^{\land})} := [\vec{x}_{j}]_{j \in \mathfrak{N}_{i}(\psi^{\land})}$, representing the stacked states of the neighbors of $i$ in the task graph $\mathfrak{G}(\psi^{\land})$.

Furthermore, for presentation, we introduce the notational equivalences: $\vec{x} \models \mathfrak{G}(\psi) \equiv \vec{x} \models \psi; \;  \vec{x} \models \mathfrak{G}(\psi_1)\mathfrak{G}(\psi_2) \ldots \equiv \vec{x} \models \psi_1 \land \psi_2 \land \ldots$, such that satisfying a \textit{sequence} of task graphs $\mathfrak{G}(\psi_1)\mathfrak{G}(\psi_2)\ldots$ corresponds to satisfy their conjunction.

%% file: content/problem.tex
Recent works \cite{Siyuan} studied decentralized feedback laws to satisfy conjunctions of tasks  $\psi^{\land}$ with the general form
\begin{equation}\label{eq:distributed law}
\vec{u}_{i}\langle \tau,\mathfrak{G}(\psi^{\land})\rangle = \kappa_i(\vec{x}_{\mathfrak{N}_{i}(\psi^{\land})}(t),t), \;\; \forall t \in [\tau, t_{\text{hr}}(\psi^{\land})]
\end{equation}
where $\kappa_i :  \mathbb{R}^{n_x \cdot |\mathfrak{N}_{i}(\psi^{\land})|} \times \mathbb{R}_{\geq 0} \rightarrow \mathbb{U}$ is a function that maps a time $t \in  [\tau, t_{\text{hr}}(\psi^{\land})]$ and the state of the neighboring agents 
$\vec{x}_{\mathfrak{N}_{i}(\psi^{\land})}$ in the task graph $\mathfrak{G}(\psi^{\land})$ to an input $\vec{u}_i$ for agent $i$ (we show shortly how we adopt \eqref{eq:distributed law} to satisfy tasks $\psi$ with disjunctions). The notation $\vec{u}_{i}\langle \tau, \mathfrak{G}(\psi^{\land})\rangle$ is used to denote that the feedback law in \eqref{eq:distributed law} is \textit{parametrized} in both the task graph $\mathfrak{G}(\psi^{\land})$ and the time $\tau \in [0,t_{\text{ac}}(\psi^{\land}))$ at which the feedback law starts to be applied to the system to satisfy $\psi^{\land}$. In practice,  $\vec{u}_{i}\langle \tau, \mathfrak{G}(\psi^{\land})\rangle: [\tau,t_{\text{hr}}(\psi^{\land})] \rightarrow\mathbb{U} $ is designed such that, starting from an initial state $\vec{x}(\tau)$, the state signal $\vec{x} : [\tau,t_{\text{hr}}(\psi^{\land})] \rightarrow \mathbb{R}^{Nn_x}$, $\vec{x}(t) = q_{\text{sol}}(t;\vec{x}(\tau),\{\vec{u}_{i}\langle \tau, \mathfrak{G}(\psi^{\land})\rangle\}_{i\in \mathcal{V}})$, satisfies the task $\psi^{\land}$, i.e., $\vec{x} \models \psi^{\land}$. Nevertheless, a key requirement to implement $\vec{u}_{i}\langle \tau, \mathfrak{G}(\psi^{\land})\rangle$ is that agent $i$ must access $\vec{x}_{\mathfrak{N}_{i}(\psi^{\land})}(t), \; \forall t\in [\tau, t_{\text{hr}}(\psi^{\land})]$, which demands $\mathfrak{N}_{i}(\psi^{\land}) \subseteq \mathcal{N}_{i,c}(t), \quad \forall t \in  [\tau, t_{\text{hr}}(\psi^{\land})],\; \forall i \in \mathcal{V}$, i.e., every neighbor in the task graph must be a neighbor in the communication graph until the time horizon $t_{\text{hr}}(\psi^{\land})$. By \textit{communication consistency} we formalize two situations for which this can not be guaranteed. 
\begin{definition}\label{def:communication consistency}
For some $t \ge 0$ and communication graph $\mathcal{G}_c(t)$, a predicate $\mu_{ij}$ with relative configuration $\mathcal{H}_{ij}(\vec{\eta}_{ij})$ is \textit{communication inconsistent} if either of the conditions
\begin{center}
\begin{tabular}{c c}
\begin{minipage}{3.0cm}
\begin{equation}\label{eq:edge inconsistent}
(i,j) \notin \mathcal{G}_c(t)
\end{equation}
\end{minipage} &
\begin{minipage}{4.2cm}
\begin{equation}\label{eq:geometric inconsistent}
\mathcal{H}_{ij}(\vec{\eta}_{ij}) \cap \mathcal{C}_{ij} = \emptyset
\end{equation}
\end{minipage}
\end{tabular}
\end{center}
hold, with $\mathcal{C}_{ij} = \{\vec{e}_{ij} \in \mathbb{R}^{n_x} \mid g^c(\vec{e}_{ij}) = \|\vec{p}_{ij}\| - r \leq 0\}$, where $r >0 $ is the communication radius. A task $\psi^{\land}$ (task graph $\mathfrak{G}(\psi^{\land})$) is communication inconsistent if it contains any such predicate (i.e., $\mu_{ij} \in \psi^{\land}$); otherwise, it is communication consistent.
\end{definition}
Informally, by Def.~\ref{def:communication consistency}, a task $\psi^{\land}$ is inconsistent at time $t$ if: 1) the agents involved in $\psi^{\land}$ are not in communication, as per \eqref{eq:edge inconsistent} (indeed, \eqref{eq:edge inconsistent} implies $\mathfrak{N}_{i}(\psi^{\land}) \not\subseteq \mathcal{N}_{i,c}(t)$ by definition of the neighbouring set), or 2) it contains a predicate $\mu_{ij}$ whose relative configuation $\mathcal{H}_{ij}(\vec{\eta}_{ij})$ is inherently infeasible with the maximum communication distance $r$. Indeed, the condition $\mathcal{H}_{ij}(\vec{\eta}_{ij}) \cap \mathcal{C}_{ij} = \emptyset$ implies that if $\mu_{ij} = \top$, i.e, $\vec{e}_{ij} \in \mathcal{H}_{ij}(\vec{\eta}_{ij})$, then necessarily $\|\vec{p}_{ij}\| > r$. Thus \eqref{eq:geometric inconsistent} implies that \eqref{eq:edge inconsistent} must hold at some time $\tau> t$, as $\mu_{ij}$ must take value $\true$ at some future time $\tau$ to satisfy $\psi^{\land}$.

Although communication consistency is time-dependent, we assume that if a task $\psi^{\land}$ is consistent at some time $t \leq t_{\text{ac}}(\psi^{\land})$, then there exists a controller that both satisfies the task and preserves consistency for all $\tau \in [t,\, t_{\text{hr}}(\psi^{\land})]$.

\begin{assumption}\label{ass:controller assamp}
    Consider a task $\psi^{\land}$, a time $\tau < t_{\text{ac}}(\psi^{\land})$, and an initial state $\vec{x}(\tau)$. Let $\psi^{\land}$ be communication consistent with $\mathcal{G}_c(\tau)$. Then there exists a control law $\vec{u}_{i}\langle \tau, \mathfrak{G}(\psi^{\land})\rangle: [\tau, t_{\text{hr}}(\psi^{\land})] \rightarrow \mathbb{U}$  as per \eqref{eq:distributed law} with corresponding solution $
    \vec{x}(t) = q_{\text{sol}}(t;\vec{x}(\tau),\{\vec{u}_{i}\langle \tau, \mathfrak{G}(\psi^{\land})\rangle\}_{i\in \mathcal{V}}), \quad \forall t \in [\tau, t_{\text{hr}}(\psi^{\land})],$
    such that $\mathfrak{G}(\psi^{\land}) \subseteq \mathcal{G}_c(t)$ for all $t \in [\tau, t_{\text{hr}}(\psi^{\land})]$, and $\vec{x} \models \psi^{\land}$.
\end{assumption}

Assmp.~\ref{ass:controller assamp} states that if $\mathfrak{G}(\psi^{\land})$ is communication consistent with $\mathcal{G}_c(\tau)$, then a controller exists that preserves the communication edges $\mathfrak{E}(\psi^{\land}) \subseteq \mathcal{E}_c(t)$ for all interval times $t\in [\tau, t_{hr}(\psi^{\land})]$, i.e., until the task is satisfied. Controllers with this property for the input affine dynamics in \eqref{eq:main dynamics} are available in prior work \cite{Siyuan,feedback1,feedback2,lindemann2025formal}, assuming $\mathfrak{G}(\psi^{\land})$ is communication consistent. Indeed, we highlight that: 1) communication maintenance could itself be represented as a relative configuration of the form \eqref{eq:predicate definition}, which can then be included within the STL task $\psi^{\land}$, 2) according to our definition, communication consistency of a configuration $\mathcal{H}_{ij}(\vec{\nu}_{ij})$ requires compatibility with the communication radius $r$ as per \eqref{eq:geometric inconsistent} in Def.~\ref{def:communication consistency}, such that the communication and relative configuration requirements within $\mathfrak{G}(\psi^{\land})$ are \textit{compatible}.

Motivated by this analysis, consider a MAS subject to a global STL task $\psi_{\text{glob}}$ as per \eqref{eq:working fragment} and connected through the dynamic communication graph $\mathcal{G}_c(t)$. We satisfy $\psi_{\text{glob}}$ under the communication constraints of the distributed law \eqref{eq:distributed law} via a two-step \textit{task decomposition} framework. First (Sec.~\ref{sec:gts}), we encode $\psi_{\text{glob}}$ as a \textit{graph transition system} (GTS), i.e., a transition system whose states are task graphs $\mathfrak{G}(\psi^{\land p_i})$, where $\psi^{\land p_i}$ is a conjunction of tasks over a common interval $I_{p_i}\in\mathcal{I}$ (Assmp.~\ref{ass:periodicity}). Satisfaction of the global task is then shown to be equivalent to satisfying a sequence of task graphs $\vec{x} \models \mathfrak{G}(\psi^{\land p_1}), \ldots, \mathfrak{G}(\psi^{\land p_{n_p}})$ along the GTS. In principle, this can be achieved by sequential application of controllers of the form \eqref{eq:distributed law}. However, individual graphs $\mathfrak{G}(\psi^{\land p_i})$ may violate communication consistency. Hence, in Sec.~\ref{sec:decomposition}, we decompose each $\mathfrak{G}(\psi^{\land p_i})$ reactively into a communication consistent task graph $\mathfrak{G}(\tilde{\psi}^{\land p_i})$ (Def.~\ref{def:communication consistency}) while preserving correctness ($\vec{x} \models \mathfrak{G}(\tilde{\psi}^{\land p_i})
\Rightarrow
\vec{x} \models \mathfrak{G}(\psi^{\land p_i})$), thus enabling the satisfaction of $\psi_{\text{glob}}$ via feedback.

%% file: content/gts.tex

In this section, we show how disjunctions in a global task $\psi_{\text{glob}}$ lead to alternative sequences of the task graphs that can satisfy $\psi_{\text{glob}}$. Namely, using the logical distributive/associative rules\footnote{$(a \land b)\lor c = (a\lor c)\land(b\lor c), \; (a \lor b)\land c = (a\land c)\lor(b\land c)$} for the $\land/\lor$ operators, we can rewrite $\psi_{\text{glob}}$ in Disjunctive Normal Form (DNF) as $\psi_{\text{glob}} = \bigvee_{k=1}^{K} \psi^{\land}_k $
for some $K \geq 1$, where the tasks $\psi^{\land}_k := \varphi_{1,k} \land \varphi_{2,k} \land \ldots$ \textit{only} contain conjunctions of tasks $\varphi$. Since each task $\varphi$ is defined over an interval $I_{p_i} \in \mathcal{I}$ (cf. Assmp. \ref{ass:periodicity}), we split each $\psi^{\land}_k$ in periods as
$
\psi^{\land}_k = \underbrace{( \varphi^{p_1}_{1,k} \land  \varphi^{p_1}_{2,k} \ldots)}_{\psi^{\land p_1}_k:=} \land  \underbrace{(\varphi^{p_2}_{1,k} \land  \varphi^{p_2}_{2,k} \ldots)}_{\psi^{\land p_2}_k:=} \land \ldots \underbrace{(\varphi^{p_{n_p}}_{1,k} \land  \varphi^{p_{n_p}}_{2,k} \ldots)}_{\psi^{\land p_{n_p}}_k:=}
$
where each task $\psi^{\land p_i}_k,\; \forall p_i = p_1, \ldots p_{n_p}$ is a conjunction of tasks $\varphi^{p_i} = T_{I_{p_i}}\theta, T\in \{G,F\}$. Thus, we write $\psi_{\text{glob}}$ as
\begin{equation}\label{eq:the transition system form}
\psi_{\text{glob}} = \bigvee\nolimits_{k=1}^{K} \left( \psi^{\land p_1}_k \land  \psi^{\land p_2}_k \land \ldots  \land \psi^{\land p_{n_p}}_k \right).
\end{equation}
The form \eqref{eq:the transition system form} clarifies that satisfaction of $\psi^{\land}_k$ corresponds to the satisfaction of a sequence of tasks $\psi^{\land p_i}_k$. Namely, recalling the semantics in \eqref{eq:semantics}, if  an index $k$ exists such that $\vec{x} \models \mathfrak{G}(\psi^{\land p_1}_k)\mathfrak{G}(\psi^{\land p_2}_k)\ldots \mathfrak{G}(\psi^{\land p_{n_p}}_k)$ then $\vec{x} \models \psi_{\text{glob}}$. Henceforth, a first approach to satisfy $\psi_{\text{glob}}$ is to directly choose one of the $K$ sequences of tasks and try to satisfy it using a feedback controller of the form \eqref{eq:distributed law}. However, such an approach would not be reactive to unexpected contingencies in the satisfaction of $\psi_{\text{glob}}$, as shown in the next example:
\begin{example}\label{ex:transition graph}
Let a global task $\psi_{\text{glob}}$ in DNF  
$\psi_{\text{glob}} =  (\psi^{\land p_1}_1 \land \psi^{\land p_2}_1) \lor (\psi^{\land p_1}_2 \land \psi^{\land p_2}_2)$
and assume $\bar{\psi}^{\land p_1}=\psi^{\land p_1}_1 = \psi^{\land p_1}_2$ for some common task $\bar{\psi}^{\land p_1}$. Equivalently we write $\psi_{\text{glob}}$ as $\psi_{\text{glob}} = \bar{\psi}^{\land p_1} \land  (\psi^{\land p_2}_1 \lor \psi^{\land p_2}_2)$. Looking at $\psi_{\text{glob}}$ using our graph intuition, then satisfaction of $\psi_{\text{glob}}$ is equivalent to first satisfying the graph $\mathfrak{G}(\bar{\psi}^{\land p_1})$ and then either $\mathfrak{G}(\psi^{\land p_2}_1)$ or $\mathfrak{G}(\psi^{\land p_2}_2)$. This shows a subtle, but important difference: before factoring $\bar{\psi}^{\land p_1}$, we had to decide a priori which of the two sequences $(\bar{\psi}^{\land p_1} \land \psi^{\land p_2}_1)$ or $(\bar{\psi}^{\land p_1} \land \psi^{\land p_2}_2)$ would be satisfied. After factorization, however, the decision between $\psi^{\land p_2}_1$ and $\psi^{\land p_2}_2$ can be postponed until after the satisfaction of the common task graph $\mathcal{G}(\bar{\psi}^{\land})$, enabling reactive selection if either $\psi^{\land p_2}_1$ or $\psi^{\land p_2}_2$ becomes infeasible (e.g., due to obstacles or agent failure).
\end{example}
Following these arguments, we introduce the notion of a graph transition system (GTS). This abstraction will enable representing sequences of task graphs satisfying $\psi_{\text{glob}}$ as \textit{runs} of a suitably defined transition system, which we compute \textit{online} rather than fix \emph{a priori} (Sec. \ref{sec:implementation}).
\begin{definition}
A graph transition system is a tuple 
$\Omega(\hat{\Gamma}_F,\{\hat{\Gamma}_{p_i}\}_{p_i = p_1}^{p_{n_p}} , \mathcal{A})$ where 
$\hat{\Gamma}_{p_i}\subset \Gamma,\; \forall p_i = p_1, \ldots p_{n_p}$, $\hat{\Gamma}_F\subset \Gamma$ is the set of \textit{terminal} graphs, and 
$\mathcal{A} : \hat{\Gamma}_{p_i} \rightarrow 2^{\hat{\Gamma}_{p_{i+1}}}$ is a transition relation. A \textit{run} $\sigma$ of $\Omega$ is a finite sequence of graphs 
$\sigma = \mathcal{G}_1 \mathcal{G}_2 \dots \mathcal{G}_L$, with 
$\sigma[k]$ denoting the $k$-th graph. A run $\sigma$ belongs to the language $\mathcal{L}_{\Omega}$, $\sigma \in \mathcal{L}_{\Omega}$, if 
$\sigma[1] \in \hat{\Gamma}_{p_1}$, 
$\sigma[L] \in \hat{\Gamma}_F$, and 
$\sigma[k+1] \in \mathcal{A}(\sigma[k])$ for all $k=1,\dots,L-1$.
\end{definition}
The GTS associated with a task $\psi_{\text{glob}}$ is constructed recursively (Alg.~\ref{alg:gts_construction}) by exploiting the distributivity of $\lor/\land$ to group tasks defined over common periods (cf. Example~\ref{ex:transition graph}). Namely, we start with a task  $\psi_{\text{glob}} = \bigvee_{k \in \mathcal{K}} 
\left( 
\psi^{\land p_1}_k \land \psi^{\land (p_i > p_1)}_k 
\right),
\quad 
\mathcal{K} = \{1,\ldots,K\},$
where for brevity 
$\psi^{\land (p_i > p_1)}_k := 
\psi^{\land p_{2}}_k \land \cdots \land \psi^{\land p_{n_p}}_k,$
and we initialize an empty automaton $\Omega$. The index set $\mathcal{K}$ is partitioned into disjoint subsets 
$\{\mathcal{K}_l\}_{l=1}^M$ (Line~\ref{alg:line:partition}) such that 
$\cup_l \mathcal{K}_l = \mathcal{K}$. 
Each subset groups indices sharing the same first-period task 
$\bar{\psi}^{\land p_1}_{\mathcal{K}_l} := \psi^{\land p_1}_k$ for all $k \in \mathcal{K}_l$ and,  
using distributivity, we have $\bigvee_{k \in \mathcal{K}_l}
\left(
\bar{\psi}^{\land p_1}_{\mathcal{K}_l}
\land
\psi^{\land (p_i > p_1)}_k
\right)
=
\bar{\psi}^{\land p_1}_{\mathcal{K}_l}
\land
\bigvee_{k \in \mathcal{K}_l}
\psi^{\land (p_i > p_1)}_k$. A graph $\mathfrak{G}(\bar{\psi}^{\land p_1}_{\mathcal{K}_l})$ is then added as a state in the GTS by adding it to $\hat{\Gamma}_{p_1}$ (Line~\ref{alg:line:add node}) and a transition is defined (Line~\ref{alg:line:add transition}) from $\mathfrak{G}(\bar{\psi}^{\land p_1}_{\mathcal{K}_l})$ toward the task graphs generated by the residual task
$\psi_{\text{next}} 
=
\bigvee_{k \in \mathcal{K}_l}
\psi^{\land (p_i > p_1)}_k,$ constructed recursively by the same procedure. The recursion proceeds until all periods are exhausted; nodes with empty transition sets are marked as terminal (Line~\ref{alg:line:add final}). Thus, solely using the distributivity property, the following lemma holds.
\begin{lemma}\label{lemma:sequence lemma}
    Consider a task $\psi_{\text{glob}}$ with GTS $\Omega$ as per Alg. \ref{alg:gts_construction}. For each run $\sigma \in \mathcal{L}_{\Omega}$ it holds that  $\vec{x} \models \sigma \Rightarrow \vec{x} \models \psi_{\text{glob}}$.
\end{lemma}
\vspace{-0.3cm}
\begin{algorithm}[H]
\caption{Construction of the GTS}
\label{alg:gts_construction}
\begin{algorithmic}[1]
\STATE \textbf{set} $\psi_{\text{glob}} = \bigvee_{k\in \mathcal{K}} 
(\psi^{\land p_1}_k \land \dots \psi^{\land p_{n_p}}_k), \mathcal{K} = \{1, \ldots K\}$
\STATE \textbf{set}  $\{\hat{\Gamma}_{p_i}\}_{p_i=p_1}^{p_{n_p}}, \hat{\Gamma}_F, \mathcal{A} \leftarrow \emptyset$

\STATE \textbf{Function} MakeNodes($\psi = \bigvee_{k \in \mathcal{K}} ( \psi^{\land p_i}_k \land  \psi^{\land (p_j > p_i)}_k )$)
\STATE Partition $\mathcal{K}$ into groups $\{\mathcal{K}_l\}$ with identical $\psi^{\land p_i}_{\mathcal{K}_l}$ \label{alg:line:partition}
\FOR{each group $\mathcal{K}_l$}
    \STATE $\hat{\Gamma}_{p_i} \leftarrow \hat{\Gamma}_{p_i} \cup \{\mathfrak{G}(\psi^{\land p_i}_{\mathcal{K}_l})\}$ \label{alg:line:add node} 
    \IF{$p_i<p_{n_p}$} \label{alg:line:next step}
        \STATE $\psi_{\text{next}} = \bigvee_{k \in \mathcal{K}_l} \psi^{\land (p_j > p_i)}_k$, 
        \STATE $\mathcal{A}(\mathfrak{G}(\psi^{\land p_i}_{\mathcal{K}_l}) \leftarrow \text{MakeNodes}(\psi_{\text{next}})$  \label{alg:line:add transition}
    \ENDIF
    \STATE \textbf{if} $\mathcal{A}(\mathfrak{G}(\psi^{\land p_i}_{\mathcal{K}_l})) = \emptyset$\; \textbf{then} \; $\hat{\Gamma}_F \leftarrow \hat{\Gamma}_F \cup \{\mathfrak{G}(\psi^{\land p_i}_{\mathcal{K}_l})\}$\label{alg:line:add final}
\ENDFOR
\RETURN $\{\mathfrak{G}(\psi^{\land p_i}_{\mathcal{K}_l})\}_{l}$
\end{algorithmic}
\end{algorithm}
In Section \ref{sec:implementation}, we will use the transition relation $\mathcal{A}$ of the GTS $\Omega$ corresponding to a global task $\psi_{\text{glob}}$ to satisfy the given task reactively (cf. Example \ref{ex:transition graph}).

%% file: content/task_dec.tex
By Lemma~\ref{lemma:sequence lemma}, we have related the satisfaction of $\psi_{\text{glob}}$ to finding a run $\sigma \in \mathcal{L}_\Omega$. As explained in greater detail in Sec. \ref{sec:implementation}, we can do so by a sequential approach by starting from the satisfaction of a task graph $\mathfrak{G}(\psi^{\land p_1}) \in \hat{\Gamma}_{p_1}$, and then a task graph $\mathfrak{G}(\psi^{\land p_2}) \in \mathcal{A}(\mathfrak{G}(\psi^{\land p_1}))$, and continue recursively to satisfy $\mathfrak{G}(\psi^{\land p_i}) \in \mathcal{A}(\mathfrak{G}(\psi^{\land p_{i-1}}))$ until $\mathfrak{G}(\psi^{\land p_i}) \in \hat{\Gamma}_{F}$. Although this procedure alone does not simplify feedback synthesis directly, it structures the problem as satisfying a sequence of task graphs. 

At the feedback level, we must ensure that task sequences are consistent with the time-varying communication graph $\mathcal{G}_c(t)$. To this end, each task $\psi^{\land p_i}$ is decomposed into a communication consistent task $\tilde{\psi}^{\land p_i}$ such that $\vec{x} \models \mathfrak{G}(\tilde{\psi}^{\land p_i}) \Rightarrow \vec{x} \models \mathfrak{G}(\psi^{\land p_i})$. At a high level, this is achieved by decomposing each inconsistent relative predicate $\mu_{ij}$ into a \textit{sequence} of relative predicates $\tilde{\vec{\mu}}_{rs}^{ij}$ over the edges $(r,s) \in \epsilon(\vec{\pi}_{ij})$ of a path $\vec{\pi}_{ij}$ connecting $i$ to $j$ in $\mathcal{G}_c(t)$ (Section~\ref{sec:task rewriting}). As we show next, we are able to \textit{compute} such predicates $\tilde{\vec{\mu}}_{rs}^{ij}$ via convex optimization (Section~\ref{sec:final optimization}). Note that we already approached this problem in detail in \cite{me}, where a more detailed analysis of the problem is found. In the presentation, we denote the set of inconsistent edges at time $t$ for $\psi^{\land p_i}$ as $\mathfrak{I}( t,\psi^{\land p_i}) :=
\left\{
(i,j)
\;\middle|\;
\exists\, \mu_{ij} \in \psi^{\land p_i}
\text{ such that }
\eqref{eq:edge inconsistent}
\text{ or }
\eqref{eq:geometric inconsistent}
\text{ hold}
\right\}$, whereas consistent edges are given by $\mathfrak{E}(\psi^{\land p_i})
\setminus
\mathfrak{I}( t,\psi^{\land p_i}).
$

\subsection{Task rewriting via parametric predicates}\label{sec:task rewriting}
Lemma \ref{lemma:main decomposition lemma} and Corollary \ref{cor:task cor} clarify how our decomposition approach is defined, using paths of connected agents in $\mathcal{G}_c(t)$.  
\begin{lemma}[\hspace{0.01cm}\cite{me}]\label{lemma:main decomposition lemma}
For time $t$ and a predicate $\mu_{ij}$ with $\mathcal{H}_{ij}(\vec{\eta}_{ij})$, let a directed path $\vec{\pi}_{ij}$ from $i$ to $j$ in $\mathcal{G}_c(t)$. Let the conjunction
\begin{equation}\label{eq:path conjunction}
\tilde{\theta}_{ij}(\vec{e}_{\epsilon(\vec{\pi}_{ij})}):= \tilde{\theta}_{ij} = \bigwedge\nolimits_{(r,s) \in \epsilon(\vec{\pi}_{ij})} \tilde{\mu}^{ij}_{rs},
\end{equation}
with,  $\vec{e}_{\epsilon(\vec{\pi}_{ij})} = [\vec{e}_{ij}]_{(i,j) \in \epsilon(\vec{\pi}_{ij})}$, and   $\tilde{\mathcal{H}}^{ij}_{rs}(\tilde{\vec{\eta}}^{ij}_{rs})$ being the formation sets for each $\tilde{\mu}^{ij}_{rs}$ as per \eqref{eq:predicate definition}, such that 
\begin{equation}\label{eq:minkowsky sum condition}
\bigoplus\nolimits_{(r,s) \in \epsilon(\vec{\pi}_{ij})} \tilde{\mathcal{H}}^{ij}_{rs}(\tilde{\vec{\eta}}^{ij}_{rs}) \subseteq \mathcal{H}_{ij}(\vec{\eta}_{ij}).
\end{equation}
Then $\tilde{\theta}_{ij} = \top  \Rightarrow \mu_{ij} = \top$.
\end{lemma}

\begin{corollary}\label{cor:task cor}
For time $t$ and period $I_{p_i}$, let $\psi^{\land p_i} = \land_{k} \varphi^{p_i}_k$, where $
\varphi_k^{p_i} = T_{I_{p_i}} \theta_k = T_{I_{p_i}}\!\left(\bigwedge_{(i,j)\in\mathcal{E}_k} \mu_{ij}\right), with 
\quad T\in\{G,F\}$, such that $\mathcal{E}_k \subseteq \mathfrak{E}(\psi^{\land p_i})$. Moreover, for all $k$, let
\begin{equation}\label{eq:rewritten equation}
\tilde{\varphi}^{p_i}_k
:= T_{I_{p_i}}(
\bigwedge_{(i,j)\in \mathcal{E}_k\setminus \mathfrak{I}( t,\psi^{\land p_i}) } \hspace{-0.3cm}\mu_{ij}
\;\land\!
\bigwedge_{(i,j)\in \mathcal{E}_k \cap \mathfrak{I}( t,\psi^{\land p_i})} \hspace{-0.3cm}\tilde{\theta}_{ij}),
\end{equation}
where predicates $\mu_{ij}, (i,j) \in \mathcal{E}_k \cap \mathfrak{I}( t,\psi^{\land p_i})$ are replaced by $\tilde{\theta}_{ij}$ as per \eqref{eq:path conjunction}.
Then,
$\vec{x} \models \tilde{\psi}^{\land p_i}
:= \bigwedge_k \tilde{\varphi}_k
\;\Rightarrow\;
\vec{x} \models \psi^{\land p_i}.
$
\end{corollary}
\begin{figure}[b]
    \centering
    \vspace{-0.2cm}
    \includegraphics[width=\linewidth]{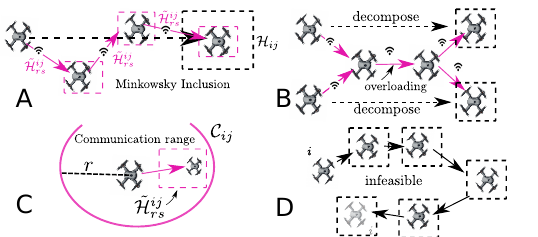}
    \caption{The four constraints required to achieve a task decomposition. (A) Minkowski inclusion constraint as per \eqref{eq:minkowsky sum condition}, (B) Task overloading when two tasks are decomposed along a common edge, (C) Communication Consistency constraint $\tilde{\mathcal{H}}_{rs}^{ij}(\tilde{\vec{\eta}}_{rs}^{ij}) \subseteq \mathcal{C}_{ij}$, (D) Cycle closure constraint.}
    \label{fig:dec figure}
\end{figure}

The results of Lemma~\ref{lemma:main decomposition lemma} and Corollary~\ref{cor:task cor} enable the construction of a task $\tilde{\psi}^{\land p_i}$ whose satisfaction implies that of $\psi^{\land p_i}$. Namely, for each inconsistent predicate $\mu_{ij}$ with $(i,j) \in \mathfrak{I}( t,\psi^{\land p_i})$, a decomposition path $\vec{\pi}_{ij}$ in $\mathcal{G}_c(t)$ is first identified. Such a path can be computed efficiently, for instance using a depth-first search with complexity $O(\mathcal{V} + \mathcal{E}_c(t))$~\cite[Sec.~29]{sedgewick2011algorithms}. In general, multiple decomposition paths may exist for a given pair $(i,j)$; this selection is discussed in Remark~\ref{remark:path selection}.
By the given paths selection, the task $\tilde{\psi}^{\land p_i}$ is obtained by replacing each $\mu_{ij}$ with the conjunction $\tilde{\theta}_{ij} = \bigwedge_{(r,s)\in \epsilon(\vec{\pi}_{ij})} \tilde{\mu}^{ij}_{rs}$, $\epsilon(\vec{\pi}_{ij}) \subseteq \mathcal{G}_c(t)$. The corresponding relative configuration sets $\tilde{\mathcal{H}}_{rs}^{ij}(\tilde{\vec{\eta}}_{rs}^{ij})$ are then selected to satisfy the Minowsky inclusion ~\eqref{eq:minkowsky sum condition} and the condition $\tilde{\mathcal{H}}_{rs}^{ij}(\tilde{\vec{\eta}}_{rs}^{ij}) \subseteq \mathcal{C}_{rs}$, where these constraints can be enforced by a set of convex constraints on the parameters $\tilde{\vec{\eta}}_{rs}^{ij}$ (see Prop \ref{prop:first prop}-\ref{prop:minkowski inclusion appendix} in Appendix for details).
By replacing all inconsistent predicates in such a manner, we obtain $\tilde{\psi}^{p_i}$ as communication consistent (cf. Def. \ref{def:communication consistency}). 

While the proposed procedure resolves the communication consistency problem, we should additionally ensure that the resulting task $\tilde{\psi}^{\land p_i}$ is satisfiable (i.e., there exists a state trajectory $\vec{x}$ such that $\vec{x}\models \tilde{\psi}^{\land p_i}$) or otherwise, no meaningful use of $\tilde{\psi}^{\land p_i}$ could be made to satisfy the communication inconsistent task $\psi^{\land p_i}$. This requires extra constraints on the selection of the relative configurations $\tilde{\mathcal{H}}_{rs}^{ij}(\tilde{\vec{\eta}}_{rs}^{ij})$ as discussed in greater detail in \cite[Sec. IV]{me}. In particular, the relative predicates $\tilde{\mathcal{H}}_{rs}^{ij}(\tilde{\vec{\eta}}_{rs}^{ij})$ should be chosen such that \textit{overloading} of edges and circular dependancies are resolved: \\
\textbf{Overloading}: When decomposing two predicates $\mu_{ij}$ and $\mu_{i'j'}$, with $(i,j),(i',j') \in \mathfrak{I}( t,\psi^{\land p_i})$, it is possible that their decomposition paths $\vec{\pi}_{ij}, \vec{\pi}_{i'j'}$ intersect, i.e., there exists $(r,s) \in \epsilon(\vec{\pi}_{ij}) \cap \epsilon(\vec{\pi}_{i'j'})$. Moreover, it is possible that a consistent predicate $\mu_{rs}$, i.e., $(r,s) \in  \mathfrak{E}(\psi^{\land p_i}) \setminus \mathfrak{I}( t,\psi^{\land p_i})$, is already present on $(r,s)$. In either case, we term the edge $(r,s)$ as \textit{overloaded}. To resolve overloading, we first ensure that only a single parametric predicate $\tilde{\mu}_{rs}^{ij}$ is defined for all paths $\vec{\pi}_{ij}$ such that $(r,s) \in \epsilon(\vec{\pi}_{ij})$. At the same time, if a consistent predicate $\mu_{rs}$ is already present on $(r,s)$, then we enforce the condition $\tilde{\mathcal{H}}_{rs}^{ij}(\tilde{\vec{\eta}}_{rs}^{ij}) \subseteq \mathcal{H}_{rs}(\vec{\eta}_{rs})$ such that, by Def. \ref{eq:predicate definition}, $\tilde{\mu}_{rs}^{ij} = \top \Rightarrow \mu_{rs} = \top$. In this way, the satisfaction of $\mu_{rs}$ is not compromised by the addition of $\tilde{\mu}_{rs}^{ij}$. Note that the inclusion $\tilde{\mathcal{H}}_{rs}^{ij}(\tilde{\vec{\eta}}_{rs}^{ij}) \subseteq \mathcal{H}_{rs}(\vec{\eta}_{rs})$ is equivalent to a set of linear inequalities in the parameters $\tilde{\vec{\eta}}_{rs}^{ij}$ (see Appendix, Cor. \ref{cor:stupic corollary}). We denote the set of overloaded edges as $\mathfrak{O}( t, \psi^{\land p_i}) = \{ (r,s) \mid (r,s) \in \epsilon(\vec{\pi}_{ij}) \land (r,s) \in \mathfrak{E}(\psi^{\land p_i}) \setminus \mathfrak{I}( t,\psi^{\land p_i}), (i,j)\in \mathfrak{I}(t , \psi^{\land p_i}) \}$. \par

\textbf{Circular Dependencies}:
Decomposing $\psi^{\land p_i}$ into $\tilde{\psi}^{\land p_i}$ requires careful examination of circular task dependencies as specified in Lemma~\ref{lemma:circular_dependency}:
\begin{lemma}( See \cite[Fact 4-5]{me})\label{lemma:circular_dependency}
Consider $\psi^{\land p_i} = \land_{k} \varphi_k^{p_i}$ and let $\vec{\pi}_{\omega}$ be a cycle in $\mathfrak{G}(\psi^{\land p_i})$ of length $L$. If at least $L-1$ predicates $\mu_{rs}, \; (r,s) \in \epsilon(\vec{\pi}_{\omega})$ are such that $\mu_{rs} \in \varphi_k^{p_i}= G_{I_{p_i}} \theta$ for some $k$, then
$\vec{0} \notin \bigoplus_{(r,s)\in\epsilon(\vec{\pi}_{\omega})} \mathcal{H}_{rs}(\vec{\eta}_{rs})\; \Rightarrow \;\vec{x}\not\models \psi^{\land p_i}$, for any state trajectory $\vec{x}: \mathbb{R}_{\geq 0} \rightarrow \mathbb{R}^{Nn_x}$. 
\end{lemma}

Henceforth, satisfying the obtained task $\tilde{\psi}^{\land p_i}$ we require that for every cycle $\vec{\pi}_{\omega}$ in $\mathfrak{G}(\tilde{\psi}^{\land p_i})$ containing at least $L-1$ tasks with an \textit{always} operator, the Minkowski sum of the relative predicates along $\vec{\pi}_{\omega}$ must contain the origin (Fig.~\ref{fig:dec figure}-D). To enforce this  we introduce auxiliary variables $\vec{\zeta}^{\mathrm{aux}}_{rs} \in \mathbb{R}^{n_x}$ for each edge $(r,s) \in \epsilon(\vec{\pi}_{\omega})$, subject to:
\begin{equation}\label{eq:cycle constraints}
    \sum_{(r,s)\in\epsilon(\vec{\pi}_{\omega})} \hspace{-0.4cm}\vec{\zeta}^{\mathrm{aux}}_{rs} = \vec{0}, \; \vec{\zeta}^{\mathrm{aux}}_{rs} \in \mathcal{H}_{rs}(\vec{\eta}_{rs}), \; \forall\,(r,s)\in\epsilon(\vec{\pi}_{\omega}),
\end{equation}
where here $\mathcal{H}_{rs}(\vec{\eta}_{rs})$ can be parametric (thus of the form $\tilde{\mathcal{H}}_{rs}^{ij}(\tilde{\vec{\eta}}_{rs}^{ij})$), or not. By definition of Minkowski sum, relation \eqref{eq:cycle constraints} certifies $\vec{0} \in \bigoplus_{(r,s)\in\epsilon(\vec{\pi}_{\omega})} \mathcal{H}_{rs}(\vec{\eta}_{rs})$, and reduce to linear inequalities in $\vec{\zeta}^{\mathrm{aux}}_{rs}$ and $\vec{\eta}_{rs}$. For presentation, we let $\mathfrak{C}(t,\psi^{\land p_i})$ be the set of cycles in $\mathfrak{G}(\tilde{\psi}^{\land p_i})$ obtained from the decomposition over the communication graph $\mathcal{G}_c(t)$.

 \subsection{Overall Optimization}\label{sec:final optimization}

In summary, in the process of rewriting $\psi^{\land p_i}$ into $\tilde{\psi}^{\land p_i}$ at time $t$, assume that the maximum length of the paths $\vec{\pi}_{ij}$ used for the decomposition is $L$. It can be shown that, to achieve the task decomposition, we introduce a number of optimization variables $M_{\text{var}}$ proportional to $M_{\text{var}} =
O\!\left(
L\left(
|\mathfrak{C}( t,\psi^{\land p_i})|\, n_x
+
|\mathfrak{I}( t,\psi^{\land p_i})|\, 2n_x
\right)
\right)$ (corresponding to the variables $\tilde{\vec{\eta}}_{rs}^{ij}$ and $\vec{\zeta}^{\text{aux}}_{rs}$), and a number of convex constraints $M_{\text{con}}$ proportional to $M_{\text{con}} = O(
(2^{n_x}\,2n_x)|\mathfrak{I}( t,\psi^{\land p_i})|
+ 2^{n_x}L|\mathfrak{I}( t,\psi^{\land p_i})| 
+ (2n_x\,2^{n_x})|\mathfrak{O}( t,\psi^{\land p_i})|
+ (n_x + 2n_x)|\mathfrak{C}( t,\psi^{\land p_i})|
)$. Consider now the vector $\vec{\xi} \in \mathbb{R}^{M_{\text{var}}}$ stacking all the decomposition variables and let $\mathfrak{g}\langle t,\psi^{\land p_i}\rangle  :  \mathbb{R}^{M_{\text{var}}} \rightarrow  \mathbb{R}^{M_{\text{con}}}$ be the function stacking all the convex inequality constraints required to decompose $\psi^{\land p_i}$ into $\tilde{\psi}^{\land p_i}$ (where the notation $\langle t, \psi^{\land p_i}\rangle$ applies to denote parametric dependence on $t$ and $\psi^{\land p_i}$). Then each vector $\vec{\xi}$ such that $\mathfrak{g}\langle t,\psi^{\land p_i}\rangle (\vec{\xi}) \leq \vec{0}$ corresponds to a valid decomposition of $\psi^{\land p_i}$ into $\tilde{\psi}^{\land p_i}$ at time $t$, for which $\vec{x} \models \tilde{\psi}^{\land p_i}\Rightarrow \vec{x} \models \psi^{\land p_i}$, and $\tilde{\psi}^{\land p_i}$ is communication consistent. In particular, given the system state $\vec{x}(t)$ at time $t$, we look for solutions $\vec{\xi}^{*}$ minimizing the convex cost $
\mathfrak{f}\langle t,\psi^{\land p_i}\rangle(\vec{\xi}) = \sum_{(i,j) \in \mathfrak{I}( t, \psi^{\land p_i}) }\sum_{(r,s) \in \epsilon(\vec{\pi}_{ij})} \| \vec{e}_{rs}(t) - \tilde{\vec{c}}_{rs}^{ij} \|^2$ which represents the sum of the squared distances of the relative vectors of each new relative configuration $\tilde{\mathcal{H}}_{rs}(\tilde{\vec{\eta}}^{ij}_{rs}), \tilde{\vec{\eta}}_{rs}^{ij} = [\tilde{\vec{c}}_{rs}^{ij},\tilde{\vec{\alpha}}_{rs}^{ij}], \; \forall (r,s) \in \epsilon(\vec{\pi}_{ij}),\; \forall (i,j) \in \mathfrak{I}( t, \psi^{\land p_i}) $ from the current relative positions $\vec{e}_{rs}(t)$. Informally,  $\mathfrak{f}\langle t,\psi^{\land p_i}\rangle(\vec{\xi})$ incentivizes the desired relative vectors of the relative configurations $\tilde{\mathcal{H}}_{rs}^{ij}(\tilde{\vec{\eta}}_{rs}^{ij})$ in $\tilde{\psi}^{\land p_i}$, (introduced by the decomposition) to be as close as possible to the actual relative state $\vec{e}_{ij}(t)$ of the agents at time $t$. To conclude, decomposing $\psi^{\land p_i}$ into $\tilde{\psi}^{\land p_i}$ is achieved by solving the convex program :
\begin{equation}\label{eq:optimizaiton final}
    \mathcal{D}\langle t, \psi^{\land p_i}\rangle: \min_{\vec{\xi}} \;  \mathfrak{f}\langle t, \psi^{\land p_i}\rangle(\vec{\xi})\; \text{s.t.} \;\mathfrak{g}\langle t, \psi^{\land p_i}\rangle(\vec{\xi}) \leq \vec{0},
\end{equation}
and we say that $\psi^{\land p_i}$ is \textit{decomposable} if $\mathcal{D}\langle t, \psi^{\land p_i}\rangle$ has a solution, and,  with an abuse of notation, we write $\tilde{\psi}^{\land p_i} = \mathcal{D}\langle t, \psi^{\land p_i}\rangle$ when $\tilde{\psi}^{\land p_i}$ is obtained by decomposing $\psi^{\land p_i}$ at time $t$ over the communication graph $\mathcal{G}_c(t)$ by solving \eqref{eq:optimizaiton final}.

\begin{remark}\label{remark:path selection}
Feasibility of $\mathcal{D}\langle t, \psi^{\land p_i}\rangle$ depends on the choice of decomposition paths $\vec{\pi}_{ij}$. Longer paths introduce more constraints (e.g., overloading and circular dependencies), but must be sufficiently long to enable decomposition. In particular, for a relative configuration $\mathcal{H}_{ij}(\vec{\eta}_{ij})$, $\vec{\eta}_{ij} = [\vec{c}_{ij},\vec{\alpha}_{ij}]$, decomposed over the relative configurations $\tilde{\mathcal{H}}_{rs}^{ij}(\tilde{\vec{\eta}}_{rs}^{ij})$, the path length $L$ should satisfy $L \geq \left\lceil \frac{\|\vec{c}_{ij}\|}{r} \right\rceil$ to ensure feasibility of the Minkowski constraint \eqref{eq:minkowsky sum condition}, jointly with with the communication constraint $\tilde{\mathcal{H}}_{rs}^{ij}(\tilde{\vec{\eta}}_{rs}^{ij}) \subseteq \mathcal{C}_{rs},\; \forall (r,s)\in \epsilon(\vec{\pi}_{ij})$.
\end{remark}

    
    
    
    
    
    

%% file: content/implementation.tex
Based on the GTS construction and the task decomposition procedure, the algorithm used to satisfy the global task $\psi_{\text{glob}}$ is presented in Alg.~\ref{alg:final algorithm}. At a high level, the proposed method alternates between \emph{task decomposition} and \emph{continuous-time execution} of a feedback controller associated with a selected task graph. Specifically, the algorithm is initialized at $t=0$ with the global task $\psi_{\text{glob}}$, its corresponding GTS $\Omega(\hat{\Gamma}_F,\{\hat{\Gamma}_{p_i}\}_{p_i=p_1}^{p_{n_p}}, \mathcal{A})$, the initial state of the system $\vec{x}_0$, and the communication graph available at the initial time $\mathcal{G}_c(t)$ (Lines \ref{alg:line:init alg}-\ref{alg:line:init gts}). By the definition of the GTS $\Omega$, a set of candidate task graphs that may be executed at each period is denoted by $\hat{\Gamma}_{\text{on}}$ and initialized as $\hat{\Gamma}_{\text{on}} = \hat{\Gamma}_{p_1}$. We use the set $\hat{\Gamma}_{\text{on}}$ to find a $\sigma \in \mathcal{L}_{\Omega}$ in the GTS \textit{online} by following the transition relation $\mathcal{A}$ of $\Omega$ starting by $\hat{\Gamma}_{\text{on}}= \hat{\Gamma}_{p_1}$. Namely, at the first iteration, the algorithm searches for a task $\psi^{\land p_1}$ whose corresponding task graph $\mathfrak{G}(\psi^{\land p_1})$ belongs to the admissible candidate set $\hat{\Gamma}_{\text{on}}$ and is decomposable with respect to the current communication graph $\mathcal{G}_c(t)$ (Line \ref{alg:line:search}). When such a task graph is found, the associated decomposed task graph $\mathfrak{G}(\tilde{\psi}^{\land p_1}) = \mathcal{D}\langle t, \psi^{\land p_1}\rangle$
is selected for execution. Recalling that for any task $\psi^{\land p_i}$ with interval $I_{p_i} = [a_{p_i}, b_{p_i}]$, the activation time and horizon are given by $t_{\text{ac}}(\psi^{\land p_i}) = a_{p_i}, t_{\text{end}}(\psi^{\land p_i}) = b_{p_i}$, then a feedback controller  $\vec{u}_{i}\langle t,\mathfrak{G}(\tilde{\psi}^{\land p_1})\rangle :
[t, b_{p_1}] \rightarrow \mathbb{U}$ is implemented on the system, satisfying Assmp.~\ref{ass:controller assamp} (Line \ref{alg:line:execution}). Under this controller, the resulting state trajectory
$\vec{x}_{p_1} : [t, b_{p_1}] \rightarrow \mathbb{R}^{N n_x}$ with $\vec{x}_{p_1}(\tau)=q_{\text{sol}}\!\left(\tau;\vec{x}_0,\{\vec{u}_{i}\langle \tau,\mathfrak{G}(\tilde{\psi}^{\land p_1})\rangle\}_{i\in \mathcal{V}}\right)\; \forall \tau \in [t, b_{p_1}]$ satisfies $\tilde{\psi}^{\land p_1}$, and thus, by Corollary \ref{cor:task cor}, also satisfies $\psi^{\land p_1}$. Moreover, by Assmp.~\ref{ass:controller assamp}, the graph $\mathfrak{G}(\tilde{\psi}^{\land p_1})$ remains communication consistent with $\mathcal{G}_c(\tau)$ over the entire interval $\tau \in [t,b_{p_1}]$. Once the execution interval ends at time $\tau=b_{p_1}$, the algorithm updates the set of candidate task graphs that can be executed next as
$\hat{\Gamma}_{\text{on}} = \mathcal{A}(\mathfrak{G}(\psi^{\land p_1})) \subseteq \hat{\Gamma}_{p_2}$. Using this updated set, the algorithm repeats the same operations accomplished for the first period. Namely, a task graph $\mathfrak{G}(\psi^{\land p_2})$ decomposable over the current communication graph $\mathcal{G}_c(b_{p_1})$ is searched for and a trajectory $\vec{x}_{p_2}:[b_{p_1}, b_{p_2}] \rightarrow \mathbb{R}^{Nn_x}$ is computed, under the feedback law  $\vec{u}_{i}\langle b_{p_1},\mathfrak{G}(\tilde{\psi}^{\land p_2})\rangle : [b_{p_1}, b_{p_2}] \rightarrow \mathbb{U}$, with $\mathfrak{G}(\tilde{\psi}^{\land p_2}) = \mathcal{D}\langle b_{p_1}, \mathfrak{G}(\tilde{\psi}^{\land p_2})\rangle$, such that $\vec{x}_{p_2} \models \psi^{\land p_2}$ and we update $\hat{\Gamma}_{\text{on}} = \mathcal{A}(\mathfrak{G}(\psi^{\land p_2})) \subseteq \hat{\Gamma}_{p_3}$. This procedure is repeated sequentially for each period. Upon reaching a terminal graph in the GTS, the algorithm terminates, yielding a satisfying trajectory formed by the concatenation of the trajectories generated over each execution interval. By construction, this trajectory satisfies the sequence of task graphs $\sigma = \mathfrak{G}(\psi^{\land p_1}) \, \mathfrak{G}(\psi^{\land p_2}) \, \mathfrak{G}(\psi^{\land p_3}) \, \ldots$
in \(\mathcal{L}_{\Omega}\). 
\begin{theorem}
Consider a global task $\psi_{\text{glob}}$, the initial time $t_0=0$ and initial state $\vec{x}_0$. The closed-loop trajectory $\vec{x}: [t_0,b_{n_p}] \rightarrow \mathbb{R}^{Nn_x}$ obtained by Alg.~\ref{alg:final algorithm} through the concatenation of closed-loop trajectories
$\vec{x}(t) = \vec{x}_{p_i}(t), \; \forall t \in [b_{p_{i-1}}, b_{p_i}], i = 1,\dots,n_p,$ with $b_{p_0} := t_0$, is such that $
\vec{x} \models \psi_{\text{glob}}.$
\end{theorem}
Henceforth, Alg. \ref{alg:final algorithm} provides a solution for feedback-based satisfaction of $\psi_{\text{glob}}$ under time-varying communication constraints, as per our problem formulation.
\begin{remark}
In Alg. \ref{alg:final algorithm}, we assumed at least one decomposable task $\mathfrak{G}(\psi^{\land p_i})$ exists in $\hat{\Gamma}_{\text{on}}$ per period. Infeasibility occurs when disconnected graphs $\mathcal{G}_c(t)$ lack decomposition paths $\vec{\pi}_{ij}$ for inconsistent predicates, or when excessive circular dependencies in $\mathfrak{G}(\tilde{\psi}^{\land p_i})$ render \eqref{eq:optimizaiton final} unsolvable. To mitigate this, agents can gather toward a common state to improve $\mathcal{G}_c(t)$ connectivity. Alternatively, for strict real-time requirements, we relax \eqref{eq:optimizaiton final} to $\mathfrak{g}\langle t, \psi^{\land p_i}\rangle(\vec{\xi}) \leq \vec{\lambda}$ for some $\vec{\lambda} \in \mathbb{R}^{M_{\text{con}}}_{\geq 0}$ and compute a minimum-violating solution.
\end{remark}

\begin{algorithm}[H]
\caption{}
\label{alg:final algorithm}
\begin{algorithmic}[1]
\STATE \textbf{set} $\text{finished} \leftarrow \text{False}$, $t \leftarrow 0$, $\vec{x}_0 \rightarrow \vec{x}$ \label{alg:line:init alg}
\STATE \textbf{set} $\Omega \leftarrow GTS(\psi_{\text{glob}})$, $\hat{\Gamma}_{\text{on}} \leftarrow \hat{\Gamma}_{p_1}$ \label{alg:line:init gts}
\WHILE{ \textbf{not} finished}
     \FOR{$\mathfrak{G}(\psi^{\land p_i})$ in $\hat{\Gamma}_{\text{on}}$}  \label{alg:line:search}
        \STATE Attempt $\mathfrak{G}(\tilde{\psi}_k^{\land p_i}) = \mathcal{D} \langle t,\tilde{\psi}_k^{\land p_i}\rangle$ using \eqref{eq:optimizaiton final} 
        \IF{problem is feasible}
             \STATE $\vec{x}_{p_i}(t) = q_{\text{sol}}(\tau ;\vec{x}_0,\{\vec{u}_{i}\langle b_{p_{i-1}},\mathfrak{G}(\tilde{\psi}^{\land p_i})\rangle\}_{i\in \mathcal{V}})$ \label{alg:line:execution}
             \STATE $\vec{x}_0 \leftarrow \vec{x}_{p_i}(b_{p_{i}})$
             \STATE \textbf{break}
        \ELSE
            \STATE \textbf{continue}
        \ENDIF
     \ENDFOR
     \STATE $\hat{\Gamma}_{\text{on}} \leftarrow \mathcal{A}(\mathfrak{G}(\psi^{\land p_i}))$,
     \STATE \textbf{if} $\hat{\Gamma}_{\text{on}} = \emptyset$ \textbf{then} $\text{finished} \rightarrow True$
\ENDWHILE

\end{algorithmic}
\end{algorithm}



%% file: content/numerical_studies.tex
We evaluate our framework in terms of the scalability of the decomposition approach and hardware implementation.
\subsubsection*{\textbf{Scalability}}
To evaluate scalability, we measure the computation time required to decompose a task graph $\mathfrak{G}(\psi^{\land p_i})$ on a communication graph $\mathcal{G}_c(\mathcal{V}, \mathcal{E}_c)$ by solving \eqref{eq:optimizaiton final}, including the graph operations needed to compute the paths $\vec{\pi}_{ij}$ and cycles $\vec{\pi}_{\omega}$. The results are shown as heatmaps in Fig.~\ref{fig:scalability}. Each row corresponds to a fixed number of agents $N \in \{10,30,50\}$. For each configuration, we generate a random communication graph $\mathcal{G}_c(\mathcal{V},\mathcal{E}_c)$ with connectivity parameter $p_{\text{comm}}\in\{0.1,0.5,0.9\}$ (y-axis), selecting randomly $|\mathcal{E}_c|=\lceil p_{\text{comm}}N(N-1)\rceil$ edges. A task graph $\mathfrak{G}(\psi^{\land p_i})$ is then generated with edges $|\mathfrak{E}(\psi^{\land p_i})|=\lceil p_{\text{task}}N(N-1)\rceil$ randomly selected for parameter $p_{\text{task}}\in\{0.1,0.5,0.9\}$ (x-axis). A fraction $p_{\text{inc}}\in\{0.1,0.3,0.5\}$ of these edges is enforced communication inconsistent (columns in Fig.~\ref{fig:scalability}), while the remaining edges are selected to be communication consistent up to the available communication edges $\mathcal{E}_c$. For each pair $(p_{\text{task}},p_{\text{comm}})$, we generate 20 random instances of the decomposition problem and report the average computation time. In all experiments, the time spent on graph operations, to find decomposition paths and cycles, dominates the time required to solve \eqref{eq:optimizaiton final}, which for all settings is below $0.5s$. From Fig.~\ref{fig:scalability}, the computational time is mainly driven by task connectivity. When both communication and task connectivity are high, computation is dominated by the enumeration of cycles $\vec{\pi}_{\omega}$ required to avoid infeasible tasks (Sec.~\ref{sec:task rewriting}). The percentage of inconsistent tasks has a minor effect at high communication connectivity, but the impact is more pronounced at low communication connectivity, where longer decomposition paths are required to decompose inconsistent predicates in the task graph.
\begin{figure}
    \centering
    \includegraphics[width=0.9\linewidth]{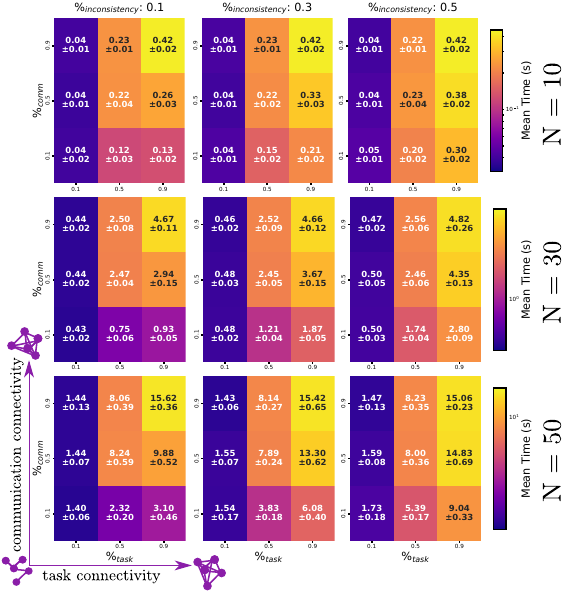}
    \caption{Scalability test. Computational time is reported in seconds. Experiments were conducted on a system equipped with an Intel Core i7-1265U processor and 32 GB of RAM.}
    \label{fig:scalability}
    \vspace{-0.3cm}
\end{figure}
\subsubsection*{\textbf{Hardware experiments}}
We validated the framework proposed in Alg. \ref{alg:final algorithm} through laboratory experiments\footnote{Video: \url{https://youtu.be/qLjHKBQRmYw}; Code: \url{https://github.com/bjarnemoro/crazyflie_ros}, \url{https://github.com/gregoriomarchesini/faste_decomposition}} using a swarm of seven Crazyflie drones integrated with ROS2 using the \textit{Crazyswarm} library \cite{crazyswarm}. The experiments demonstrate the agents successfully achieving a global task $\psi_{\text{glob}}$ as represented by the GTS in Fig. \ref{fig:gts} and leveraging Alg. \ref{alg:final algorithm}. In Fig. \ref{fig:gts}, each node signifies a possible task graph $\mathfrak{G}(\psi^{\land p_i})$, where $\psi^{\land p_i}$ is the conjunction of the individual tasks $\varphi$ within that node. The sequence $\sigma$ selected to satisfy $\psi_{\text{glob}}$ is highlighted in red. For brevity in Fig. \ref{fig:gts} we wrote $\\|\vec{e}_{ij} - \vec{c}_{ij}\|_{\infty} \leq \alpha$ to denote the relative configuration set $\mathcal{H}_{ij}(\vec{\eta}_{ij}) = \{ \vec{e}_{ij}\mid H(\vec{e}_{ij} - \vec{c}_{ij}) - M \vec{1}\alpha\}$ as per \eqref{eq:predicate level set}.

\usetikzlibrary{matrix, arrows.meta}

\begin{figure}
\centering
\resizebox{1.\columnwidth}{!}{
\begin{tikzpicture}[>=Stealth, thick]

\tikzset{
  taskbox/.style={
    draw, 
    rectangle, 
    rounded corners=1pt, 
    inner sep=4pt, 
    align=left, 
    fill=white,
    font=\scriptsize
  },
  labelnode/.style={
    font=\small\bfseries,
    inner sep=2pt
  }
}

\matrix (m) [matrix of nodes, 
             row sep=2.5em, 
             column sep=1em, 
             nodes={anchor=center},
             ampersand replacement=\&] 
{
  |[labelnode]| $I_{p_1} = [25,30]$ \& 
    |[taskbox, alias=root,draw=red]| {
      $\begin{aligned}
        &\varphi = G_{I}\,\|e_{1,11}-[2,0]\|_\infty \leq 0.2 \\
        &\varphi = G_{I}\,\|e_{1,3}-[1,1.4]\|_\infty \leq 0.2 \\
        &\varphi = G_{I}\,\|e_{11,3}-[-1,1.4]\|_\infty \leq 0.1 \\
        &\varphi = G_{I}\,\|x_1-[0,0]\|_\infty \leq 0.1
      \end{aligned}$
    } \& \\ 
  |[labelnode]| $I_{p_2} = [75,80]$ \& 
    |[taskbox, alias=p2spec,draw=red]| {
      $\begin{aligned}
        &\varphi = F_{I}\,\|e_{1,6}-[3.4,0]\|_\infty \leq 0.2\\
        &\varphi = G_{I}\,\|x_1-[0,0]\|_\infty \leq 0.2
      \end{aligned}$
    } \&
    |[taskbox, alias=p2alt]| {
      $\begin{aligned}
        &\varphi = F_{I}\,\|e_{3,6}-[3.4,0]\|_\infty \leq 0.2\\
        &\varphi = G_{I}\,\|x_3-[0,0]\|_\infty \leq 0.2
      \end{aligned}$
    } \\ 
  |[labelnode]| $I_{p_3} = [140,160]$ \& 
    |[taskbox, alias=p3spec,,draw=red]| {
      $\begin{aligned}
        &\varphi = G_{I}\,\|e_{6,4}-[0,1]\|_\infty \leq 0.2\\
        &\varphi = G_{I}\,\|e_{6,5}-[0,-1]\|_\infty \leq 0.2\\
        &\varphi = G_{I}\,\|e_{6,11}-[-1.7,0.4]\|_\infty \leq 0.1\\
        &\varphi = F_{I}\,\|e_{6,10}-[-1.7,-0.4]\|_\infty \leq 0.1\\
        &\varphi = F_{I}\,\|x_6-[0.5,0]\|_\infty \leq 0.2
      \end{aligned}$
    } \&
    |[taskbox, alias=p3alt]| {
      $\begin{aligned}
        &\varphi = G_{I}\,\|e_{4,6}-[0,1]\|_\infty \leq 0.2\\
        &\varphi = G_{I}\,\|e_{5,6}-[0,-1]\|_\infty \leq 0.2\\
        &\varphi = G_{I}\,\|e_{11,6}-[-1.7,0.4]\|_\infty \leq 0.1\\
        &\varphi = F_{I}\,\|e_{10,6}-[-1.7,-0.4]\|_\infty \leq 0.1\\
        &\varphi = F_{I}\,\|x_4-[0.5,0]\|_\infty \leq 0.2
      \end{aligned}$
    } \\ 
};

\draw[->] (root.south) -- (p2spec.north);
\draw[->] (root.south) -- (p2alt.north);
\draw[->] (p2spec.south) -- (p3spec.north);
\draw[->] (p2alt.south) -- (p3alt.north);

\end{tikzpicture}
}
\caption{Graph Transition System $\Omega$ for the experiments.}
  \vspace{-0.3cm}
\label{fig:gts}
\end{figure}
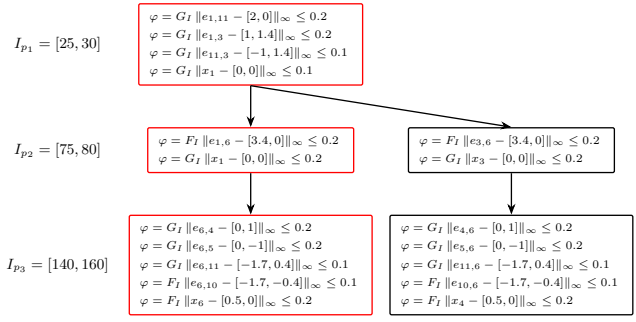

%% file: content/conclusion.tex
We proposed a communication-aware STL decomposition framework for multi-agent formations. Using a graph transition system and task redistribution, the approach maintains decentralized feedback-based reactivity under time-varying communication. Scalability across agent count and task complexity was verified in simulation and hardware experiments. Future work includes scaling experiments to larger swarms.

%% file: content/appendix.tex
In this section, we show that: 1) the inclusion $\mathcal{H}_{ij}(\vec{\eta}) \subseteq \mathcal{C}_{ij}$ in a convex set $\mathcal{C}_{ij}$ (e.g., the communication set in Def. \ref{def:communication consistency}) is equivalent to a set of convex constraints on $\vec{\eta}$ (Prop. \ref{prop:first prop}-Cor. \ref{cor:stupic corollary}) 2) the Minkowski inclusion condition $\bigoplus\nolimits_{(r,s) \in \epsilon(\vec{\pi}_{ij})} \mathcal{H}_{rs}(\vec{\eta}_{rs}) \subseteq \mathcal{H}_{ij}(\vec{\eta}_{ij})$ (appearing in \eqref{eq:minkowsky sum condition}) is equivalent to a set of linear inequalities in $\vec{\eta}_{rs},\; \forall (r,s) \in \epsilon(\vec{\pi}_{ij})$ (Prop. \ref{prop:minkowski inclusion appendix}). 
\begin{proposition}\label{prop:first prop}
Let a relative configuration $\mathcal{H}_{ij}(\vec{\eta}_{ij})$ and a set $\mathcal{C}_{ij} := \{ \vec{e}_{ij} \in \mathbb{R}^{n_x} \mid g_k(\vec{e}_{ij}) \leq 0, \;  k = 1, \ldots n_g \}$, where $n_g \geq 1$ and each $g_k: \mathbb{R}^{n_x} \rightarrow \mathbb{R}$ is convex. Then $\mathcal{H}_{ij}(\vec{\eta}_{ij}) \subseteq \mathcal{C}_{ij}$ if an only if 
\begin{equation}\label{eq:equivalence convex inclusion}
  g_k(v_l(\vec{\eta}_{ij})) \leq 0, \; \forall l=1,\ldots 2^{n_x}\; ,\; \forall k = 1, \ldots n_g.
\end{equation}
where each $v_l : \mathbb{R}^{2n_x} \rightarrow \mathbb{R}^{n_x},\; l = 1, \ldots 2^{n_x}$ takes the form 
\begin{equation}\label{eq:vertex equation}
    v_l(\vec{\eta}_{ij}) = \vec{c}_{ij} + \frac{\vec{\alpha}_{ij}}{2} \odot \vec{s}_l, \quad \vec{s}_l \in \{-1,1\}^{n_x}.
\end{equation}
\end{proposition}
\begin{proof}(Sketch)
By definition, the set $\mathcal{H}_{ij}(\vec{\eta}_{ij})$ can be written as the \textit{convex hull} \cite[Sec. 2.2.4]{boyd2004convex} of a set of $n_v = 2^{n_x}$ \textit{vertices} $v_l(\vec{\eta}_{ij})$ as $\mathcal{H}_{ij}(\vec{\eta}_{ij}) = \{\sum_{l=1}^{n_v} \lambda_l v_l(\vec{\eta}_{ij}) \mid\sum_{l=1}^{n_v} \lambda_l = 1, \lambda_l \geq 0  \}$, where each vertex $v_l, l =1, \ldots 2^{n_x},$ is a linear function of $\vec{\eta}_{ij}$ as per \eqref{eq:vertex equation}. By  Jensen's inequality, \cite[Sec. 3.1.8]{boyd2004convex}, we can then conclude that $\mathcal{H}_{ij}(\vec{\eta}_{ij}) \subseteq \mathcal{C}_{ij}$ if and only if each vertex $\vec{v}_l(\vec{\eta}_{ij})$ satisfies the inequality $g_k(v_l(\vec{\eta}_{ij})) \leq 0\; \forall k = 1, \ldots n_g$, which is exactly \eqref{eq:equivalence convex inclusion}. 
\end{proof}
\begin{corollary}\label{cor:stupic corollary}
Given relative configurations $\mathcal{H}_{ij}(\vec{\eta}_{ij})$ and $\mathcal{H}'_{ij}(\vec{\eta}'_{ij})$, then $\mathcal{H}_{ij}(\vec{\eta}_{ij}) \subseteq \mathcal{H}'_{ij}(\vec{\eta}'_{ij})$ if an only if $g_{k}(v_l(\vec{\eta}_{ij})) := \vec{h}_k^T(v_l(\vec{\eta}_{ij}) - \vec{c}'_{ij}) - \vec{m}_k^T\vec{\alpha}_{ij}' \leq 0,
           \; \forall l=1,\ldots 2^{n_x}\;,
           \forall k = 1, \ldots 2n_x$, where $\vec{h}_k$ and $\vec{m}_k$ are the $k$-th row of the matrices $H = [\text{Id}_{n_x}, -\text{Id}_{n_x}]^T$, $M = [\text{Id}_{n_x},\text{Id}_{n_x}]^T$.
\end{corollary}
\begin{proof}
Since $\mathcal{H}'_{ij}(\vec{\eta}'_{ij}) = \{\vec{e}_{ij} \mid g_k(\vec{e}_{ij})= \vec{h}_k^T (\vec{e}_{ij} - \vec{c}'_{ij}) - \vec{m}_k^T\vec{\alpha}_{ij}' \leq 0, \forall k = 1, \ldots 2n_x\}$, where $g_k$ is linear in $\vec{e}_{ij}$, and thus convex, then the result follows by Prop. \ref{prop:first prop}. 
\end{proof}

\begin{proposition}\label{prop:minkowski inclusion appendix}
Let a graph $\mathcal G$, and a path $\vec{\pi}_{ij}$ in $\mathcal{G}$, where $\vec{e}_{rs} \in \mathcal{H}_{rs}(\vec{\eta}_{rs}),\; \forall (r,s)\in \epsilon(\vec{\pi}_{ij})$, for some relative configuration $\mathcal{H}_{rs}(\vec{\eta}_{rs})$. Let the stacked vector $\bar{\vec{\eta}} = [\vec{\eta}_{rs}]_{(r,s) \in \epsilon(\vec{\pi}_{ij})}$  and a relative configuration $\mathcal{H}_{ij}(\vec{\eta}_{ij})$. Then, the inclusion condition $\bigoplus_{(r,s) \in \epsilon(\vec{\pi}_{ij})} \mathcal{H}_{rs}(\vec{\eta}_{rs}) \subseteq \mathcal{H}_{ij}(\vec{\eta}_{ij})$ is equivalent to a set of linear inequalities in $\bar{\vec{\eta}}$. 
\end{proposition}

\begin{proof}(Sketch)
By \eqref{eq:minkowsky inclusion to linear constraints},  $\bigoplus_{(r,s) \in \epsilon(\vec{\pi}_{ij})} \mathcal{H}_{rs}(\vec{\eta}_{rs}) = \mathcal{H}_{ij}\left(\sum_{(r,s) \in \epsilon(\vec{\pi}_{ij})} \vec{\eta}_{rs}\right) = \mathcal{H}_{ij}(\vec{\eta}^{\oplus}_{ij}(\bar{\vec{\eta}}))$, where $\vec{\eta}^{\oplus}_{ij}(\bar{\vec{\eta}}) := \sum_{(r,s) \in \epsilon(\vec{\pi}_{ij})} \vec{\eta}_{rs}$ is a linear in each $\vec{\eta}_{rs}$. By virtue of Corollary \ref{cor:stupic corollary}, we then have that $\mathcal{H}_{ij}(\vec{\eta}^{\oplus}_{ij}) \subseteq \mathcal{H}_{ij}(\vec{\eta}_{ij}) \Leftrightarrow g_{k}(\bar{\vec{\eta}}) = \vec{h}_k^T(v_l(\vec{\eta}^{\oplus}_{ij}(\bar{\vec{\eta}})) - \vec{c}_{ij}) - \vec{m}_k^T\vec{\alpha}_{ij} \leq 0,
\; \forall l=1,\ldots 2^{n_x}\;,\forall k = 1, \ldots 2n_x$. Since the function $v_l(\vec{\eta}^{\oplus}_{ij}(\bar{\vec{\eta}}))$ is the composition of linear function $v_{l}(\cdot)$ with the linear function $\vec{\eta}^{\oplus}_{ij}(\cdot)$, then the inequalities are linear in the variables $\bar{\vec{\eta}}$. 
\end{proof}